\documentclass[conference]{IEEEtran}

\IEEEoverridecommandlockouts    

\usepackage{cite}
\usepackage{graphicx}
\usepackage{amsfonts}
\usepackage{amsmath} 
\usepackage{amssymb}

\usepackage{amsthm}	
\usepackage{dsfont}
\usepackage{enumerate}
\usepackage{xcolor}
\usepackage{comment}
\usepackage{subfig}
\usepackage{pifont}
\usepackage{array}

\newtheorem{thm}{Theorem}

\newtheorem{defn}{Definition}

\renewcommand{\arraystretch}{0.9}

\newcommand{\Rset}{\mathbb{R}}

\newcommand{\Uset}{\mathbb{U}}

\newcommand{\Xset}{\mathbb{X}}

\newcounter{l1}
\newcounter{l2}
\newcounter{l3}
\newcommand{\bdotlist}{\begin{list}{$\bullet$}{}}
\newcommand{\bboxlist}{\begin{list}{$\Box$}{}}
\newcommand{\bbboxlist}{\begin{list}{\raisebox{.005in}{{\tiny
$\blacksquare$ \ \ }}}{}}
\newcommand{\bdashlist}{\begin{list}{$-$}{} }
\newcommand{\blist}{\begin{list}{}{} }
\newcommand{\barablist}{\begin{list}{\arabic{l1}}{\usecounter{l1}}}
\newcommand{\balphlist}{\begin{list}{(\alph{l2})}{\usecounter{l2}}}
\newcommand{\bAlphlist}{\begin{list}{\Alph{l2}.}{\usecounter{l2}}}
\newcommand{\bdiamlist}{\begin{list}{$\diamond$}{}}
\newcommand{\bromalist}{\begin{list}{(\roman{l3})}{\usecounter{l3}}}

\def\ts{\Delta}
\def\leak{\lambda}
\def\etaIn{\eta^{c}}
\def\etaOut{\eta^{ d}}
\def\tikzcheckmark{\checkmark} 
\def\tikzxmark{\scriptsize \ding{53}} 
\def\paper{paper}

\title{When are Lossy Energy Storage Optimization Models Convex?\\
\thanks{
The authors are with the School of Electrical and Computer Engineering, Cornell University, Ithaca, NY, 14853, USA.  Emails:  Feras Al Taha (foa6@cornell.edu), Eilyan Bitar (eyb5@cornell.edu).   
This work was supported in part by a postgraduate fellowship from the Natural Sciences and Engineering Research Council of Canada and in part by the Cornell Engineering Sprout Awards program. 
The authors thank Professor Tyrone Vincent for  helpful discussions during the course of this research.}%
}

\author{\IEEEauthorblockN{Feras Al Taha, \textit{Student Member, IEEE,} \,  Eilyan Bitar, \textit{Member, IEEE}}}

\begin{document}

\maketitle
\thispagestyle{empty}
\pagestyle{empty}

\begin{abstract}
     We examine a class of optimization problems involving the optimal operation of a single lossy energy storage system, where energy losses occur during charging and discharging. These inefficiencies typically lead to a nonconvex set of feasible charging and discharging power profiles. In this paper, we derive an equivalent reformulation of this class of optimization problems by eliminating the charging and discharging power variables and recasting the problem entirely in terms of the storage state-of-charge variables. We show that the feasible set of the proposed reformulation is always convex. We also provide sufficient conditions under which the objective function of the proposed reformulation is guaranteed to be convex. The conditions provided both unify and generalize many existing conditions for convexity in the literature. 
\end{abstract}


\section{Introduction} \label{sec:introduction}

With the surge in growth of wind and solar energy sources in the power grid, energy storage systems have become increasingly important tools to manage the variability in power supplied by these intermittent renewable energy resources. To effectively optimize the planning and operation of an energy storage system, one must account for the potential dissipation of energy while charging or discharging the storage system. However, the treatment of such nonidealities in the energy storage model can lead to nonconvex optimization problems.

A commonly used model that accounts for energy losses in the storage dynamics involves representing the storage charging and  discharging powers using distinct decision variables. While this representation results in a linear model for the energy storage dynamics, one must enforce additional ``complementarity constraints'' to prevent the occurrence of simultaneous charging and discharging. These complementarity constraints can be expressed using either bilinear equality constraints \cite{elsaadany2023battery} or binary variables \cite{chen2013mpc,jabr2014robust,parisio2014model,erdinc2014smart,wu2015stochastic,hemmati2017technical,rostamiyan2023miqp}, both of which give rise to nonconvex optimization problems. 

To side-step this nonconvexity, one can simply drop the complementarity constraints, resulting in a convex outer approximation (i.e., a superset) of the original nonconvex feasible set \cite{haessig2021convex}. 
This simple convex relaxation can be strengthened by incorporating additional valid convex inequalities  \cite{pozo2022linear, pozo2023convex, elsaadany2023battery}.
However, a crucial drawback of these convex relaxations is that they frequently yield solutions that are infeasible for the original problem, i.e.,  solutions where charging and discharging occur simultaneously.
Alternatively, to prevent these infeasibilities, one can construct a convex inner approximation (i.e., a subset) of the original feasible set \cite{nazir2021guaranteeing,elsaadany2023battery}.
While these inner approximations are guaranteed to only contain feasible solutions, a shortcoming is that they may exclude optimal solutions to the original problem.

To address the problem of infeasibility without introducing additional constraints that restrict the feasible set, there are a number of papers in the literature that provide sufficient conditions for the exactness of the convex relaxation obtained solely by dropping the complementarity constraints.
Such conditions guarantee that optimal solutions to the relaxed optimization problem never involve simultaneous charging and discharging, ensuring they are both feasible and optimal for the original nonconvex optimization problem.
For example, the convex relaxation has been shown to be exact for ``copper plate’’ economic dispatch problems with lossy energy storage when the total cost of power generation at each stage is an increasing linear function \cite{garifi2018control} or an increasing  convex quadratic function \cite{wu2016distributed}. 

Several authors have also investigated the exactness of such relaxations when applied to transmission-constrained economic dispatch and unit commitment problems with lossy energy storage \cite{castillo2013profit,li2015sufficient,duan2016improved,li2018extended,garifi2020convex,chen2022battery}. 
To guarantee exactness of the convex relaxation, these papers provide conditions that depend explicitly on the optimal solution of the relaxed optimization problem. 
A crucial drawback of such conditions is that they can only be verified a posteriori, that is, after having solved the convex relaxation itself. 
As an exception, Lin \textit{et al.} \cite{lin2023relaxing} provide conditions for exactness that can be verified a priori, but impose an additional restriction on the energy storage system's initial and maximum state-of-charge. 

\subsection{Main contributions}

In this \paper, we take the perspective of a storage owner-operator, and study a class of optimization problems involving the optimal operation of a single lossy energy storage system in the absence of power transmission constraints. Optimization problems of this kind arise in a variety of applications, including demand response, peak shaving, frequency regulation, and energy price arbitrage.

To begin, we show that there exists a piecewise affine bijection between the set of feasible power profiles and the set of feasible energy (i.e., state-of-charge) profiles associated with a lossy energy storage system. 
Using this bijection, we eliminate the power profile with a substitution of variables. This yields an equivalent reformulation of the original optimization problem, where, in the reformulation, the energy profile is the sole decision variable. 
Crucially, the set of feasible energy profiles is shown to be a convex polytope, which is explicitly expressed in half-space representation. 
In contrast to the various convex inner and outer approximations that have been proposed in the literature, the convex reformulation of the feasible set proposed in this \paper \ is exact (all feasible solutions are preserved without introducing infeasible solutions).

We also derive sufficient conditions on the objective function of the original optimization problem that guarantee convexity of the objective function of the reformulated optimization problem. 
These sufficient conditions for convexity  (which can be checked a priori) both unify and generalize a number of existing conditions for convexity in the literature, which either require smoothness of the objective function \cite{li2018extended,garifi2020convex}, assume that the objective function is piecewise affine \cite{garifi2018control,hashmi2019optimal}, or impose restrictive conditions on the energy storage system's initial and maximum state-of-charge \cite{lin2023relaxing}. 
Moreover,  existing conditions for convexity require the objective function to be nondecreasing in both the storage charging and discharging power variables \cite{li2018extended,garifi2020convex, lin2023relaxing}. In contrast,  the convexity of our reformulated objective function only requires the original objective function to be nondecreasing in the storage charging power variables, without requiring monotonicity with respect to the storage discharging power variables.  

\subsubsection*{Organization}
The remainder of this \paper \ is organized as follows. 
In Section \ref{sec:formulation}, we present the class of lossy energy storage optimization models considered in this paper.
In Section \ref{sec:convex_reformulation}, we provide an equivalent reformulation for this class of optimization problems, along with sufficient conditions for their convexity.
Section \ref{sec:conclusion} concludes the~\paper.

\subsubsection*{Notation} 
Denote by $\Rset$ the set of real numbers, and by $\Rset_+$  the set of nonnegative real numbers. 
Given two vectors $x \in \Rset^n$ and $y \in \Rset^m$, we let $(x, \, y) \in \Rset^{n+m}$ denote the vector formed by stacking $x$ and $y$ vertically.  
Given a vector $x\in\Rset^n$, we let  $x^+ := \max\{0,x\}$ and $x^-:=\min \{0,x\}$  represent the element-wise positive and negative parts of the vector $x$, respectively.
Given two vectors $x,y\in\Rset^n$, the relation $x \succeq y$ means that $x-y \in \Rset_+^n$.

\section{Problem Formulation} \label{sec:formulation}
In this section, we present a conventional model for the optimal operation of a lossy energy storage system with energy leakage and conversion inefficiencies.

We consider a discrete-time model of the energy storage system (ESS) dynamics, with time periods of duration $\ts > 0$, indexed by $t=0,1, \dots, T$. We let $u_t \in \Rset$ represent the ESS charging power level during time period $t$. A positive charging power level represents the injection of energy into the ESS, while a negative charging power level represents the withdrawal of energy from the ESS. 
We let $x_t \in \Rset$ represent the energy level (or state-of-charge) of the ESS at the beginning of time period $t$. The ESS state-of-charge evolves according to 
\begin{align} \label{eq:dyn}
    \hspace{-.05in} x_{t+1} = \leak x_t + \ts \Big ( \etaIn u_t^+ + \frac1\etaOut u_t^- \Big ), \quad t= 0, \dots, T-1,
\end{align}
where the initial state-of-charge $x_0 \in \Rset$ is assumed to be given. The scalar parameters  $\etaIn \in (0, 1]$ and $\etaOut \in (0, 1]$ represent the ESS  charging and discharging efficiencies, respectively, and $\leak \in (0, 1]$ represents the self-discharge rate of the ESS. The dynamics specified in Eq. \eqref{eq:dyn} are nonlinear, because of the need to model the dissipation of energy due to charging and discharging separately.\footnote{The lossy energy storage model specified in \eqref{eq:dyn} is equivalent to the commonly used model in the literature which employs separate decision variables to represent charging and discharging power. Although such models result in linear storage dynamics, they require additional nonconvex complementarity constraints to disallow the simultaneous charging and discharging of the ESS.} 

The ESS energy and power variables must also satisfy the following constraints:
\begin{align}
     \label{eq:energy_con}   \underline{x}_t &\leq x_t \leq \overline{x}_t, \quad t=1, \dots, T\\
   \label{eq:power_con} -\underline{u}_t &\leq u_t \leq \overline{u}_t, \quad t= 0, \dots, T-1. 
\end{align}
Here, the parameter $\overline{x}_t \geq 0$ represents the maximum amount of energy that can be stored at time $t$, $\underline{x}_t \ge 0$ represents the smallest amount of energy that can be stored at time $t$, $\overline{u}_t  \geq 0$ represents the maximum charging power at time $t$, and  $\underline{u}_t \geq 0$ represents the maximum discharging power at time $t$. 

Before presenting the optimization model considered in this \paper, it will be helpful to express the dynamics in \eqref{eq:dyn} in terms of a nonlinear mapping $\phi: \Rset^T \rightarrow \Rset^T$ from the ESS \emph{power profile} $u: = (u_0, \dots, u_{T-1})$ to the corresponding \emph{energy profile} $x:= (x_1, \dots, x_T)$, which is given by
\begin{align} \label{eq:phi_map}
    x = \phi(u) := A \left( \etaIn u^+ + \frac1\etaOut u^-\right) + b.
\end{align}
Here, $b := (\leak x_0, \, \leak^2 x_0, \, \dots, \, \leak^T x_0) \in \Rset^T$ and $A\in\Rset^{T\times T}$ is a lower triangular matrix with nonnegative entries  given by $A_{ij}:=\ts \leak^{i-j}$ for all $j\le i$.
It is also important to note that the matrix $A$ is invertible since it is lower triangular and all of the elements along its diagonal are nonzero. Its inverse is given by $A^{-1} = (1/\Delta)(I - \leak L)$, where $I \in \Rset^{T \times T}$ is the identity matrix and $L \in \Rset^{T \times T}$ is the lower shift matrix, i.e., a matrix with ones along its subdiagonal and zeros everywhere else.
Using the relationship in Eq. \eqref{eq:phi_map}, we can express the family of ESS optimization problems studied in this \paper \ as
\begin{align}
\nonumber \text{minimize}    &   \  \ c(u) \\
\label{eq:opt} \text{subject to}  &  \ \ x = \phi(u)\\
\nonumber & \ \ \underline{x} \preceq x \preceq \overline{x}\\
\nonumber & \hspace{-.077in} - \underline{u} \preceq u \preceq \overline{u}.
\end{align}
Here, the decision variables are  $x\in \Rset^T$ and $u \in \Rset^T$, and the objective function  $c: \Rset^T \rightarrow \Rset$ is assumed to be convex. The vectors $\overline{u}:=(\overline{u}_0,\dots,\overline{u}_{T-1})$ and $\underline{u}:=(\underline{u}_0,\dots, \underline{u}_{T-1})$ encode the maximum charging and discharging limits on the power profile, respectively,  and the vectors $\overline{x}:=(\overline{x}_1,\dots,\overline{x}_T)$ and  $\underline{x}:=(\underline{x}_1,\dots,\underline{x}_T)$ encode the upper and lower limits on the energy storage profile,~respectively.  

A power profile $u \in \Rset^T$ is said to be feasible for problem~\eqref{eq:opt} if it induces an energy profile $x = \phi(u)$ such that both profiles $(x,u)$  satisfy constraints \eqref{eq:energy_con}-\eqref{eq:power_con}. 
Formally, one can express the \emph{set of feasible power profiles} as 
\begin{align} \label{eq:uset}
    \Uset := \left \{u\in\Rset^T \, | \,  u \in [-\underline{u}, \, \overline{u}], \ \, \phi(u) \in [\underline{x}, \, \overline{x}] \right\}.
\end{align}
It is important to recognize that the nonlinear equality constraint $x  = \phi(u)$ in problem \eqref{eq:opt} can lead to a nonconvex set of feasible power profiles. We illustrate this point in  Fig. \ref{fig:uset}, where we provide a simple example of a lossy energy storage system with a nonconvex set of feasible power profiles.  

It is also important to note that the set of feasible power profiles $\Uset$ is guaranteed to be convex if the energy storage system is lossless (i.e., $\etaIn=\etaOut=1$), or restricted to charging only (i.e., $u_t\ge 0$ for all stages $t$) or discharging only (i.e., $u_t \le 0$ for all stages $t$).

\begin{figure}[t!]
    \centering
    \subfloat[Set of feasible power profiles $\Uset$]{\label{fig:uset} \includegraphics[width=0.53\columnwidth]{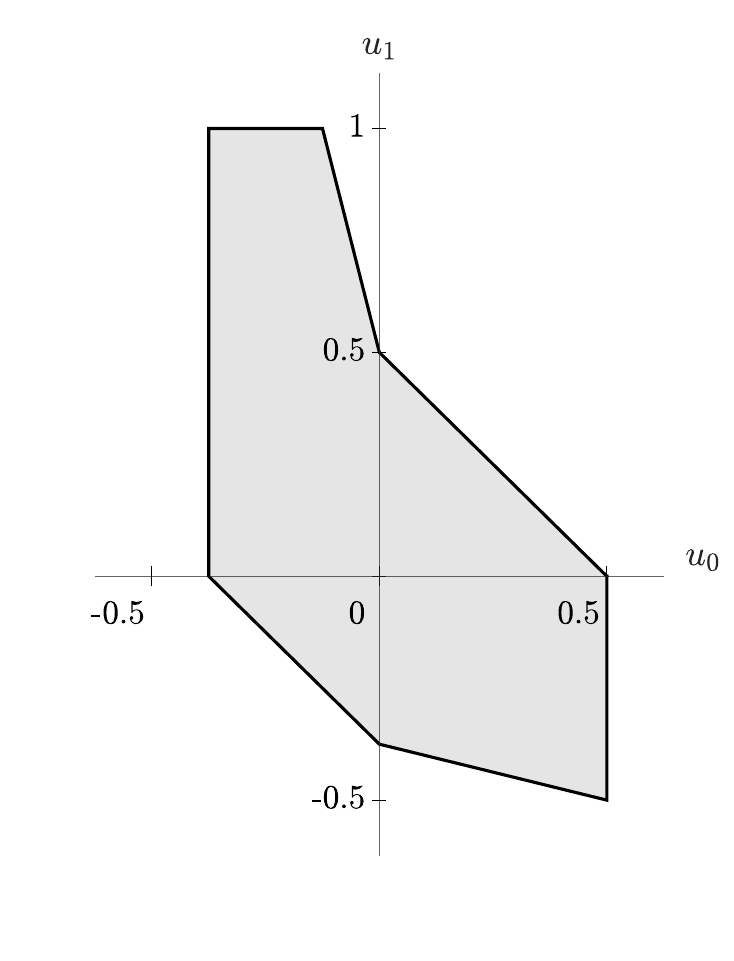}}\quad 
    \subfloat[Set of feasible energy profiles $\Xset=\phi(\Uset)$]{\label{fig:xset} \includegraphics[width=0.53\columnwidth]{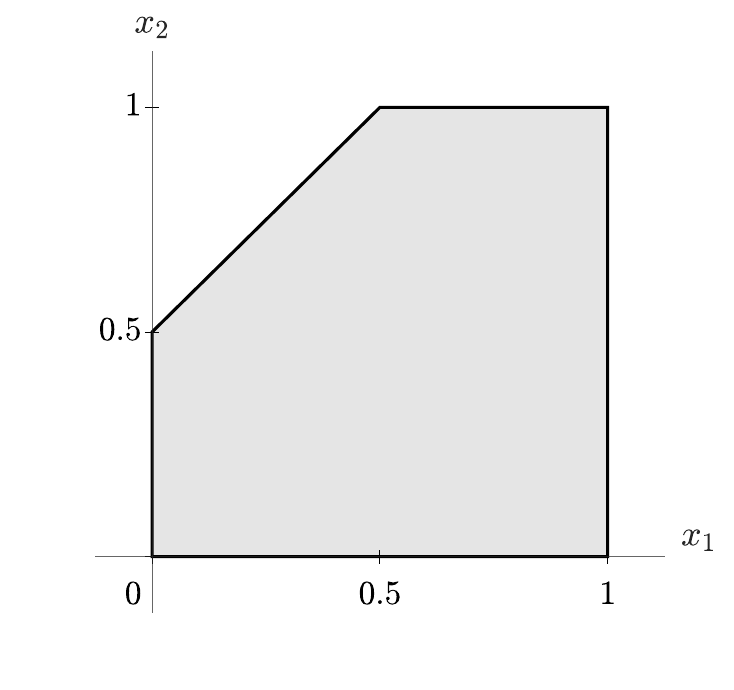}\quad}
    \caption{Depiction of (a) the set of feasible power profiles $\Uset$ and (b) the set of feasible energy profiles $\Xset=\phi(\Uset)$ for a simple two-period ($T=2$) lossy energy storage system model with parameters: $\etaIn=\etaOut=0.5$, $\ts=1$, $x_0=0.75$, $\leak=1$, $\underline{u}=(1,1)$, $\overline{u}=(1,1)$, $\underline{x}=(0,0)$, and $\overline{x}=(1,1)$.}
    \label{fig:feas_set}
\end{figure}

\section{Convex Reformulation} \label{sec:convex_reformulation}

In this section, we show that the mapping $\phi$ from power to energy profiles is a piecewise affine bijection. 
This key property allows us to recast the original problem \eqref{eq:opt} solely in terms of energy profiles.  Importantly, we show that, unlike the potentially nonconvex set of feasible power profiles, the set of feasible energy profiles is always a convex set. We also provide sufficient conditions on the objective function that ensure its convexity under the proposed reformulation of problem \eqref{eq:opt}.
The following result will play a key role in the derivation of these results.

\begin{thm} \label{lem:phi} \rm
    The function $\phi: \Rset^T \rightarrow \Rset^T$, as defined in Eq. \eqref{eq:phi_map}, is concave and invertible.\footnote{When referring to convexity (resp. concavity) of a vector-valued function, we mean that each of its scalar-valued component functions is a convex (resp. concave) function.}  Furthermore, its inverse, denoted by $\phi^{-1}$, is a convex function  given by
    \begin{align} \label{eq:phi_inv}
        \phi^{-1}(x) =  \frac{1}{\etaIn} \left(A^{-1}(x-b) \right)^+ + \etaOut \left(A^{-1}(x-b) \right)^-
    \end{align}  
    for all $x \in \Rset^T$.
\end{thm}

\begin{proof} It will be helpful to express the function $\phi$ as a composition of functions given by $\phi(u) =  Af(u) + b$ , where $$f(u) := \etaIn u^+ + (1/\etaOut) u^-$$ for all $u \in \Rset^T$. Using the fact that $0 < \eta^c \leq 1/\eta^d$, it can be shown that the function $f$ is both concave and invertible, where the inverse of $f$ is given by
\begin{align*}
    f^{-1}(v) := (1/\etaIn) v^+ + \etaOut v^-
\end{align*}
for all $v \in \Rset^T$ and can be shown to be convex. 

Recall that the lower triangular matrix $A$ is entrywise nonnegative, since $A_{ij}=\ts \leak^{i-j} > 0$ for all $j\le i$. It follows that each component of the vector-valued function $\phi(u) = Af(u) + b$ is a nonnegative linear combination of concave functions (the components of $f$) plus an offset $b$. Thus,  $\phi$ is a concave function. 

To derive the inverse of $\phi$, let $u \in \Rset^T$ and define $x =  Af(u) + b$. Since both the matrix $A$ and the function $f$ are invertible, we have that
\begin{align*}
    u &= f^{-1}\left(A^{-1}(x-b)\right)\\
        &= \frac1\etaIn \left(A^{-1}(x -b)\right)^+ + \etaOut \left(A^{-1} (x -b)\right)^- \\
      &= \phi^{-1}(x).
\end{align*}

To complete the proof, note that each component of the vector-valued function $\phi^{-1}(x) = f^{-1}(A^{-1}(x-b))$ is given by a composition of a convex function with an affine function. Thus, $\phi^{-1}$ is a convex function.
\end{proof}

Given the invertibility of the map $\phi$, we can apply the substitution $u = \phi^{-1}(x)$ to reformulate problem \eqref{eq:opt} such that the energy profile $x$ becomes the sole decision variable. 
The resulting \emph{set of feasible energy profiles} can be expressed~as
\begin{align} \label{eq:xset}
    \Xset := \left \{x\in\Rset^T \, | \, \phi^{-1}(x) \in [-\underline{u}, \, \overline{u}], \  x \in [\underline{x}, \, \overline{x}] \right \}.
\end{align}

At first glance, the set $\Xset$ may appear to be nonconvex  due to the inequality constraint $   \phi^{-1}(x) \succeq - \underline{u}$ (recall from Thm.~\ref{lem:phi} that $\phi^{-1}$ is a convex function). However, the following result reveals that the set of feasible energy profiles is in fact a convex polytope by providing an equivalent representation of $\Xset$ as an intersection of half-spaces.
In Fig.~\ref{fig:xset}, we provide a simple example illustrating the convexity of the set $\Xset$.

\begin{thm} \rm \label{thm:xset} 
    The set of feasible energy profiles $\Xset$, as defined in Eq. \eqref{eq:xset}, is a convex polytope that can be expressed in half-space representation as
    \begin{align} \label{eq:xset_cvx}
        \Xset = \big\{x\in\Rset^T \, \big| \,  A^{-1} (x -b) \in \big[-{\textstyle \frac1\etaOut}\, \underline{u}, \, \etaIn \, \overline{u}\big], \ x \in [\underline{x},\, \overline{x}] \big\}.
    \end{align}
\end{thm}
\begin{proof} Let $x \in [\underline{x}, \, \overline{x}]$.
    To prove \eqref{eq:xset_cvx}, it suffices to show that 
    \begin{align*}
        -\underline{u} \preceq \phi^{-1}(x) \preceq \overline{u} \ \ \Longleftrightarrow \ \,
        -\frac1\etaOut\underline{u} \preceq A^{-1}(x-b) \preceq \etaIn\overline{u}.
    \end{align*}
    Recall the expression for $\phi^{-1}$ in Eq. \eqref{eq:phi_inv}.
Since $\overline{u} \succeq 0$, it holds that $ \phi^{-1}(x) \preceq \overline{u} $ if and only if $(1/\etaIn)  A^{-1}(x -b)  \preceq \overline{u}$. Since $-\underline{u} \preceq 0$, it also holds that $ \phi^{-1}(x) \succeq -\underline{u}  $ if and only if  $\etaOut  A^{-1}(x -b) \succeq -\underline{u} $. The desired result follows, since $\etaIn, \etaOut > 0$. 
\end{proof}

Applying the substitution $u = \phi^{-1}(x)$ and using the half-space representation of the set of feasible energy profiles  given in Eq. \eqref{eq:xset_cvx}, one can reformulate the original optimization problem \eqref{eq:opt} as
\begin{align}
\nonumber \text{minimize}    &   \  \ c(\phi^{-1}(x)) \\
\label{eq:opt_reform} \text{subject to}  &  \ \ \underline{x} \preceq x \preceq \overline{x},\\
\nonumber & \hspace{-.077in} - \frac1\etaOut\underline{u} \preceq A^{-1}(x-b) \preceq \etaIn\overline{u},
\end{align}
where the energy profile $x \in \Rset^T$ is the sole decision variable. 

\newcounter{rowcount}
\setcounter{rowcount}{0}
\begin{table*}[ht!]
    \centering
    \begingroup
    \setlength{\tabcolsep}{6pt} 
    \renewcommand{\arraystretch}{2.6} 
    \begin{tabular}{@{\stepcounter{rowcount}(\alph{rowcount})\hspace*{\tabcolsep}}lllcc}
    \hline \hline 
    \multicolumn{1}{>{\makebox[-4pt][r]{}}l}{ \textbf{Application}}
    &  \textbf{Cost function $c$}  & \textbf{Parameters} & \textbf{Assumptions}   & \textbf{Convexity of $c\circ \phi^{-1}$}\\ 
    \hline 
    Peak shaving    & ${ \max_{t\in\{0,\dots,T-1\}} |u_t+\ell_t|}$    & Load: $\ell_t \in \Rset$ & $\ell_t \ge 0, \ \forall  t $       & \tikzcheckmark \\  \hline 
    Load balancing     & $ \sum_{t=0}^{T-1} (u_t+\ell_t)^2$    &  Load: $\ell_t \in \Rset$ & $\ell_t \ge 0, \ \forall t $       & \tikzcheckmark \\  \hline  
    Power regulation     & $ \sum_{t=0}^{T-1} |u_t - r_t|$    & Regulation signal: $r_t \in \Rset$ &  $r_t \leq 0, \ \forall t$ &  \tikzcheckmark\\  \hline 
    Energy arbitrage     & $ \sum_{t=0}^{T-1} p_t^{\rm buy} u_t^+ + p_t^{\rm sell} u_t^-$    & Prices: $p_t^{\rm buy}, p_t^{\rm sell}\in\Rset$ & $   (1/\etaIn) p_t^{\rm buy} \ge \etaOut p_t^{\rm sell} , \ \forall t$       & \tikzcheckmark \\  \hline  
    Power smoothing     & $ \sum_{t=1}^{T-1} |(s_{t}-u_{t})-( s_{t-1}- u_{t-1})|$    & Renewable power: $s_t \in \Rset$  &  ---   & \tikzxmark \\  \hline
    \hline
    \end{tabular}
    \endgroup
    \caption{This table presents several examples of lossy energy storage optimization problems along with their associated cost function $c$. The convexity of the objective function $c \circ \phi^{-1}$ under the  reformulation \eqref{eq:opt_reform} is determined for each example.}
    \label{tab:examples}
\end{table*}

While the feasible set of the proposed reformulation \eqref{eq:opt_reform} is guaranteed to be a convex set, the reformulated objective function $c \circ \phi^{-1}$ may be nonconvex. In \cite{hashmi2019optimal}, the composite function $c \circ \phi^{-1}$ is shown to be convex if the original cost function $c$ belongs to a specific class of convex piecewise linear functions that arise in net energy metering applications.
In what follows, we characterize a more general class of convex cost functions $c$ that preserve  convexity of the composition $c \circ \phi^{-1}$. To begin, we introduce the following definition.
 
\begin{defn} \rm
    A function $f: \mathbb{R}^n \rightarrow \mathbb{R}$ is defined to be   \emph{coordinate-wise nondecreasing  on a set} $\mathcal{S} \subseteq \mathbb{R}$ if  for each coordinate $i \in \{1, \dots, n\}$ and any two elements $x, y \in \mathbb{R}^n$  with $x_i, y_i \in \mathcal{S}$, $x_i \leq y_i$, and $x_j = y_j$ for all $j \neq i$, we have $f(x) \leq f(y)$. 
\end{defn}
In words, a multivariate function is coordinate-wise nondecreasing on a subset of the real number line if the function is nondecreasing in each variable on that subset.

\begin{thm} \label{thm:cvx_cond} \rm 
If $c:\mathbb{R}^T \rightarrow \mathbb{R}$ is convex and coordinate-wise nondecreasing on $[0,\infty)$, then    $c \circ \phi^{-1}$ is convex.
\end{thm}

\begin{proof}
To streamline the proof, we introduce a  function $g:\Rset^T \rightarrow \mathbb{R}^T$, defined   as   $g(x) := (1/\etaIn) x^+ + \etaOut x^-$ for all $x \in \Rset^T$.  It can be shown that the function $g$ is convex, since $0<\etaOut \le 1/\etaIn$. Using this newly defined function, we have that $(c \circ \phi^{-1})(x) = (c \circ g)(A^{-1}(x-b))$ for all $x \in \Rset^T$.
Thus, to prove that $c \circ \phi^{-1}$ is convex, it suffices to show that the function $c \circ g$ is convex (since the composition of a convex function with an affine function is also convex). 

We begin by proving that   $c \circ g$ is convex for the special case of $T=1$. Let $x, y \in \Rset$, $\theta  \in [0, \, 1]$, and define $z := \theta x + (1-\theta)y$. 
If $z \geq 0$, then 
\begin{align*}
    (c\circ g)(z) & \leq c(   \theta g(x)  + (1- \theta) g(y))\\
    & \leq \theta c(g(x)) + (1- \theta) c(g(y)).
\end{align*}
Since $z \geq 0$, it holds that $g(z) = (1/\etaIn) z \geq 0$. 
This fact, together with the assumption that $c$ is nondecreasing on $[0, \infty)$ and the convexity of $g$, implies the first inequality. 
The second inequality follows from the convexity of  $c$.
If  $z \leq 0$, then
\begin{align*}
    (c\circ g)(z) = c(\etaOut z) & \leq \theta c(\etaOut x) + (1- \theta) c(\etaOut y) \\
    & \leq \theta c(g(x)) + (1- \theta) c(g(y)).
\end{align*}
The first equality follows from $z\leq 0$. 
The first inequality follows from the convexity of $c$. 
To see why the second inequality is true, consider the term $c(\etaOut x)$.
If $x\le 0$, then $c(\etaOut x) = c(g(x))$. Conversely, if $x\ge 0$, then $c(\etaOut x) \le c((1/\etaIn) x) = c(g(x))$, since $c$ is nondecreasing on $[0, \infty)$ and $0 < \etaOut\le 1/\etaIn$.
Hence, $c(\etaOut x) \le c(g(x))$ for all $x\in\Rset$.  It also holds that $c(\etaOut y) \le c(g(y))$ for all $y\in\Rset$, yielding the second inequality. The case of $T > 1$ is handled similarly, where the same arguments can be applied in a ``coordinate-wise'' fashion to establish the convexity of $c \circ g$. 
\end{proof}

To ensure convexity of the composite function $c\circ \phi^{-1}$,
 Thm.~\ref{thm:cvx_cond}  requires the function $c$ to be nondecreasing with respect to the power \emph{injected} into the energy storage system at every time period. 
This sufficient condition generalizes existing results that require, more strongly, that the objective function be monotonic in both the charging and discharging power variables \cite{li2018extended,garifi2020convex, lin2023relaxing}.
We note that Thm. \ref{thm:cvx_cond} can also be viewed as an extension of standard conditions for convexity of composite functions.
For example, the conditions provided in \cite[Chapter 3.2.4]{boyd2004convex} require that the function $c$ be nondecreasing in each argument over the \emph{entire} real number line to ensure convexity of $c\circ\phi^{-1}$.

\subsection{Discussion}

While there are many objective functions of practical interest that satisfy the sufficient conditions for convexity in Thm.~\ref{thm:cvx_cond}, the proposed reformulation \eqref{eq:opt_reform} is not a panacea. 
In Table \ref{tab:examples}, we provide a variety of examples where the composite function $c \circ \phi^{-1}$ is guaranteed to be convex.  We also include an example where convexity is not guaranteed.

In examples (a) through (c), the given cost function $c$ is convex in the storage power profile $u$. 
The additional assumptions stated alongside these examples also ensure that the cost function $c$ is coordinate-wise nondecreasing on the set $[0, \infty)$. 
Hence, the resulting composite function $c\circ\phi^{-1}$ is convex, as guaranteed by Thm. \ref{thm:cvx_cond}.
It is important to point out that examples (a) through (c) do not satisfy the conditions for convexity provided in \cite{li2018extended,garifi2020convex, lin2023relaxing}, which require the cost function $c$ to be coordinate-wise nonincreasing on $(-\infty,0]$ (i.e., nonincreasing in the discharging power variables). 

In the energy arbitrage example (d), the given cost function $c$ is convex and coordinate-wise nondecreasing on $[0, \infty)$ if the energy prices satisfy the inequality $p_t^{\rm buy} \geq (p_t^{\rm sell})^+$ for every stage $t$.  However, it can be shown that the composite function $c \circ \phi^{-1}$ is convex if the energy prices satisfy a weaker condition given by $(1/\etaIn) p_t^{\rm buy} \ge \etaOut p_t^{\rm sell}$ for every stage $t$. Interestingly, the family of cost functions satisfying this condition includes cost functions $c$ that are \emph{not} coordinate-wise nondecreasing on $[0, \infty)$. This reveals that the sufficient conditions for convexity provided in Thm.~\ref{thm:cvx_cond} are not necessary in general. 

\section{Conclusion} \label{sec:conclusion}
In this \paper, we have examined a class of optimization problems involving the optimal operation of a single lossy energy storage system in the absence of power transmission constraints.  We have provided an equivalent reformulation for this class of problems, along with sufficient conditions for the convexity of the proposed reformulation.

\bibliographystyle{IEEEtran}
\bibliography{references}{\markboth{References}{References}}

\end{document}